\documentclass[]{article}
\usepackage{PRIMEarxiv}
\usepackage[utf8]{inputenc} 
\usepackage[T1]{fontenc}    
\usepackage{hyperref}       
\usepackage{url}            
\usepackage{booktabs}       
\usepackage{amsfonts}       
\usepackage{nicefrac}       
\usepackage{microtype}      
\usepackage{lipsum}
\usepackage{fancyhdr}       
\usepackage{graphicx}       
\graphicspath{{media/}}     
\usepackage{xcolor} 
\usepackage{amsthm}
\usepackage{amsmath}
\usepackage{bbold}
\usepackage{algorithm}
\usepackage[noend]{algpseudocode}
\usepackage{tikz}
\usetikzlibrary{trees}
\usetikzlibrary{positioning, arrows.meta} 

\newtheorem{definition}{Definition}
\newtheorem{theorem}{Theorem}
\newtheorem{corollary}{Corollary}
\newtheorem{observation}{Observation}

\newcommand{\err}[0]{\textnormal{err}}

\usepackage{etoolbox}
\makeatletter
\patchcmd{\@maketitle}{\LARGE}{\Large}{}{}
\makeatother
\title{{InfTDA: A Simple TopDown Mechanism for Hierarchical Differentially Private Counting Queries}}

\author{
  Fabrizio Boninsegna \\
  University of Padova \\
  \texttt{fabrizio.boninsegna@phd.unipd.it} \\
}

\begin{document}
\maketitle

\begin{abstract}
This paper extends \texttt{InfTDA}, a mechanism proposed in (Boninsegna, Silvestri, PETS 2025) for mobility datasets with origin and destination trips, in a general setting. The algorithm presented in this paper works for any dataset of $d$ categorical features and produces a differentially private \emph{synthetic} dataset that answers all \emph{hierarchical} queries, a special case of \emph{marginals}, each with bounded maximum absolute error. 
The algorithm builds upon the TopDown mechanism developed for the 2020 US Census.
\end{abstract}

\section{Introduction}
Consider a dataset $D \in \mathcal{X}^n$, where the data universe $\mathcal{X}$ is defined as the Cartesian product of $d$ categorical features, $\mathcal{X} = \mathcal{X}_1 \times \dots \times \mathcal{X}_d$. For example, this setting can represent a healthcare dataset with attributes such as ("medical condition", "nationality", "age", "gender"), or a geographic dataset with attributes such as ("country", "state", "city"). One of the most relevant classes of queries in these datasets is marginals. Given a subset $S \subseteq [d]$ and specific values $y_i \in \mathcal{X}_i$ for each $i \in S$, a marginal query asks: "how many records in $D$ have, for each attribute $i \in S$, the value equal to $y_i$?" The aim is to release a differentially private (DP) synthetic dataset $\tilde{D}$ that can accurately answer all such marginal queries. \texttt{InfTDA}, introduced in \cite{boninsegna2025} for mobility data, constructs a synthetic dataset $\tilde{D}$ that achieves a bounded maximum absolute error of $\tilde{O}(\sqrt{k^3 d})$ for any $k$-hierarchical query, where a $k$-hierarchical query is a marginal query with $S = [1, \dots, k]$. 
The algorithm employs a TopDown approach, initially developed by the US Census Bureau \cite{Abowd20222020}, which leverages differentially private and optimized $k$-hierarchical queries as constraints to refine the estimation of $(k+1)$-hierarchical queries.

\paragraph{Related Work}
The most straightforward method is to generate a \emph{contingency table} of $D$, which is its histogram representation \footnote{The histogram representation of $D$ is $h_D \in \mathbb{N}_0^{|\mathcal{X}|}$ defined for each $x \in \mathcal{X}_1\times \dots \times \mathcal{X}_d$ as $h_D(x) = |\{y\in D: x=y\}|$.}, and add Gaussian noise to each count \footnote{It can be used also Laplace noise, however here we focus in Gaussian noise as it provides better tail bounds.}. The resulting differentially private contingency table can answer all marginal queries; however, the variance grows linearly with the number of noisy counts required, leading to an absolute error of $\tilde{O}(e^{(d-k)/2})$ for any marginal with $|S| = k$.
By focusing solely on $k$-hierarchical queries, an alternative approach would be to release a differentially private contingency table $\tilde{D}_i$ for each query. In this case, each query could be answered with a maximum absolute error of $\tilde{O}(\sqrt{k d})$ \footnote{$k$ comes from a union bound on Gaussian tails over $\prod_{i=1}^{k}|\mathcal{X}_i| = O(2^k)$ elements.}. However, this strategy requires handling $d$ separate tables, which may be inconsistent with one another.
Hay et al. \cite{hierarchical} were the first to introduce the concept of hierarchical queries and to observe that accuracy can be improved by resolving such inconsistencies. The authors proposed a hierarchical tree structure, similar to the one defined for \texttt{InfTDA}, to organize the queries and developed the hierarchical mechanism, which enforces consistency through constrained optimization. A key limitation of this approach is that it produces a full DP contingency table with non-integral counts and without ensuring non-negativity. The authors identified the incorporation of non-negativity constraints as an open direction for future work. 
The most notable work with real-world application is undoubtedly the Disclosure Avoidance System TopDown algorithm, developed by the US Census Bureau \cite{Abowd20222020}. This method follows a similar principle to the hierarchical mechanism but differs in that it does not answer all $k$-hierarchical queries simultaneously. Instead, queries are answered sequentially, starting from the root of the hierarchy and proceeding down the tree. The mechanism employs a more sophisticated convex optimization procedure and produces a complete differentially private contingency table $\tilde{D}$ with non-negative and integer counts. Nevertheless, no formal utility analysis of the algorithm has been provided.
An alternative method for generating general purpose synthetic DP datasets is the Multiplicative Weights Exponential Mechanism (MWEM)~\cite{NIPS2012_208e43f0}, which is perhaps the most well-known in this category. MWEM achieves an error of $\tilde{O}(\sqrt{n d})$, thus its performance is constrained by sampling error.

\paragraph{Contributions}
\texttt{InfTDA} is a specific implementation of a hierarchical TopDown algorithm that employs constrained optimisation based on maximum error minimisation (Chebyshev distance), rather than mean squared error minimisation. The algorithm produces a differentially private contingency table $\tilde{D}$ with non-negative integer counts and guarantees a maximum absolute error of $\tilde{O}(\sqrt{k^3 d})$ for each $k$-hierarchical query. Intuitively, the proof shows that the error introduced by Gaussian noise accumulates linearly across the levels of the hierarchical tree.
The method is scalable, as it can be executed independently on each branch of the hierarchical tree, and each privatisation mechanism runs in at most $O(\max_i |\mathcal{X}_i|)$ time. Alongside \texttt{InfTDA}, Boninsegna and Silvestri \cite{Abowd20222020} introduced \texttt{IntOPT}, an optimisation procedure that operates entirely in the integer domain and requires at most $O(\max_i |\mathcal{X}_i|^2)$ steps to solve the constrained problem.

\section{Preliminaries}
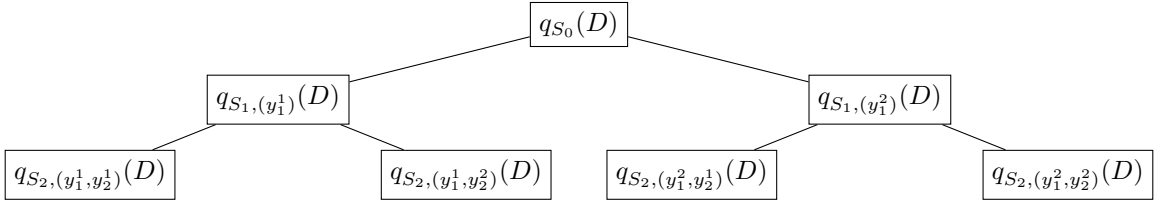
\begin{figure}[t]
\centering
\begin{tikzpicture}[
    every node/.style={rectangle, draw},
    level distance=1cm,
    level 1/.style={sibling distance=8cm}, 
    level 2/.style={sibling distance=5cm}   
]
\node (root){$q_{S_0}(D)$}
    child {node (n1) {$q_{S_1, (y_{1}^1)}(D)$}
      child {node (n11) {$q_{S_2, (y_{1}^1,y_{2}^1)}(D)$}}
      child {node (n12) {$q_{S_2, (y_{1}^1,y_{2}^2)}(D)$}}
    }
    child {node (n2) {$q_{S_1, (y_{1}^2)}(D)$}
      child {node (n21) {$q_{S_2, (y_{1}^2, y_2^1)}(D)$}}
      child {node (n22) {$q_{S_2, (y_{1}^2, y_2^2)}(D)$}}
    };
\end{tikzpicture}
\caption{\footnotesize An example of the non-negative hierarchical tree defined by hierarchical queries for a data universe $\mathcal{X} = \{y_1^1, y_1^2\} \times \{y_2^1, y_2^2\}$. The hierarchical consistency assures that the tree is hierarchical, so that child nodes' attributes sum to their father's attribute.}
\label{fig:tree}
\end{figure}

\paragraph{Non-Negative Hierarchical Tree} A tree $\mathcal{T}$ of depth $d$ is a tuple $\mathcal{T}=(V, E)$:  $V=\cup_{k=0}^{d}V_k$ is the set of nodes, with $V_k$ denoting the set of nodes at level $k \in \{0,\dots,d\}$; $E$ is the set of edges between consecutive levels. A node at level $k$ is indicated with $u_k \in V_{k}$. 
The set of children to a node is indicated by a function $\mathcal{C}: V_{k}\rightarrow 2^{V_{k+1}}$ for any level $k\in \{0,\dots, d\}$. Each node of the tree has an attribute $q(u_k)$. 
\texttt{InfTDA} works for any tree data structure defined as follows

\begin{definition}[Non-Negative Hierarchical Tree \cite{boninsegna2025}] A tree $\mathcal{T} = (V, E)$ is said to be non-negative if it contains non-negative attributes $q(u)\geq 0$ for any $u\in V$. The tree is hierarchical if the attribute of $u$ is the sum of the attributes of its children $\mathcal{C}(u)$. Specifically, for every $u\in V$ which is not a leaf, we have $q(u) = \sum_{v\in \mathcal{C}(u)}q(v)$.
\end{definition}

This paper focuses on integer attributes, so that $q(u)\in \mathbb{N}_0$ for any $u \in V$. Additionally, it will be used ${\bf q}_{\mathcal{C}}(u_k)$, which is the vector containing the attributes of the children of $u_k$.

\paragraph{Hierarchical Queries}
In this paper, we focus on data points containing $d$ categorical features $\mathcal{X} = \mathcal{X}_1\times \dots\times \mathcal{X}_d$ where each $\mathcal{X}_i$ is a set of $|\mathcal{X}_i|\geq 2$ categories. A dataset consists of a multiset of $n$ points for $\mathcal{X}$, thus $D\in \mathcal{X}^n$. 

\begin{definition}[Counting Query] A counting query is defined by a function $q_i: \mathcal{X}\rightarrow \{0,1\}$. Given a dataset $D$, we will abuse notation and write $q_i(D)$ to denote the sum of the values of the function $q_i$ on $D$ as $q_i(D) =\sum_{x\in D}q_i(x)$
\end{definition}

A $k$-way marginal query counts the number of points in $D$ that share the same values for a specified set of $k$ attributes.

\begin{definition}[$k$-way marginal query] A $k$-way marginal is defined by a subset $S\subseteq [d]$ of $|S|=k$, together with a particular value of each of the features $y \in \prod_{i\in S}\mathcal{X}_i$ with $y_i \in \mathcal{X}_i$ for any $i \in S$. Given such a pair $(S, y)$, let $\mathcal{X}(S, y)= \{x\in \mathcal{X}: x_i =y_i \,\,\forall i \in S\}$. The corresponding counting query $q_{S, y}$ is defined as $q_{S, y}(x) = \mathbb{1}(x\in \mathcal{X}(S, y))$
\end{definition}

There are $2^d$ possible $k$-way marginal queries, corresponding to the cardinality of the power set of $[d]$. This paper focuses on a specific subset of these marginal queries, referred to as \emph{k-hierarchical queries}.
\begin{definition}[$k$-hierarchical query] A $k$-hierarchical query is a $k$-way marginal query with $S = [1, \dots, k]$. 
\end{definition}
The hierarchy follows from a simple observation, \emph{a single $k$-hierarchical query can be obtained by summing all the $k+1$ hierarchical queries with the same features.}
\begin{observation}[Hierarchical Consistency] For any $k \in [d-1]$ and for any $y \in \prod_{i=1}^{k}\mathcal{X}_i$ we have that
\begin{equation*}
    q_{S_{k}, y}(D) = \sum_{y_{k+1}\in \mathcal{X}_{k+1}}q_{S_{k+1}, (y, y_{k+1})}(D).
\end{equation*}
For convenience, we define $q_{S_0}(D) = n$ so that $q_{S_0}(D) = \sum_{y_1 \in \mathcal{X}_1}q_{S_1,y_1}(D)$.
\end{observation}
Thus, all the $k$-hierarchical queries can be arranged into a non-negative hierarchical tree, as shown in Figure \ref{fig:tree}. The tree contains at each level $k$ all the $k$-hierarchical queries, thus for any $y \in \prod_{i=1}^{k}\mathcal{X}_i$ there is a unique $u_{k}\in V_k$ s.t. $q(u_k)=q_{S_k, y}(D)$. The error considered in this work is the \emph{maximum absolute error} for each $k$-hierarchical query, hence $\max_{y\in \prod_{i=1}^{k}\mathcal{X}_i}|q_{S_{k, y}}(D)-q_{S_{k, y}}(\Tilde{D})|$ for any $k \in \{0,\dots, d\}$.

\paragraph{Differential Privacy}
The neighbouring relation considered in this work is by \emph{substitution} of a single data point, which is more precisely called \emph{bounded} differential privacy \cite{kifer2011no}. With this definition, the size of the dataset does not change among neighbouring datasets, making the query $q_{S_0}(D)=n$ non-sensitive. The mechanism considered in this work is the \emph{multivariate discrete Gaussian mechanism} \cite{Canonne_Kamath_Steinke_2022}, which adds integer noise sampled from a discrete Gaussian distribution satisfying zero-concentrated differential privacy (zCDP) \cite{bun2016concentrated}.

\section{TopDown Mechanism with Chebyshev Distance}
This section introduces \texttt{InfTDA}. The algorithm constructs a privatized non-negative hierarchical tree with integer attributes, where the leaves correspond to the counts of a private full contingency table $\tilde{D}$. The mechanism employs a TopDown approach: at each level $k$, discrete Gaussian noise is added to the attributes, and they are optimized to satisfy the consistency relationship with level $k-1$, and non-negativity, by minimizing the Chebyshev distance (i.e., the maximum error) with the private estimates.

\paragraph{InfTDA}
\begin{algorithm}[t]
\centering
\caption{\texttt{InfTDA}}\label{alg: InfTDA}
\begin{minipage}[t]{0.52\textwidth}
\begin{algorithmic}[1]
\Require Tree $\mathcal{T}=(V,E)$ of depth $d$, privacy budget $\rho>0$.
\State $\tilde{V}_0\gets \{(u_0, n)\}$
\For{$k \in (1,\dots, d)$} 
    \State $\tilde{V}_k \gets \{\}$ 
    \For {$(u, c) \in \tilde{V}_{k-1}$} 
        \vspace{2pt}
        \State $\tilde{\bf q} \gets {\bf q}_{\mathcal{C}}(u)+\mathcal{N}_{\mathbb{Z}}\big(0,d/\rho\big)^{|\mathcal{X}_{k}|}$
        \vspace{2pt}
        \State $\bar{\bf q} \gets \texttt{IntOpt}\big(\tilde{\bf q}, c\big)$ 
        \vspace{2pt}
        \State $C \gets \mathcal{C}(u)$ 
        \State $X \gets \{(C_j, \bar{q}_j)\,:\, \bar{q}_{j}>0\}$ 
        \State $\tilde{V}_{k} \gets \tilde{V}_{k} \cup X$ 
    \EndFor
\EndFor\\
\Return $\tilde{D} = \tilde{V}_{d}$
\end{algorithmic}
\end{minipage}
\hfill
\begin{minipage}[t]{0.46\textwidth}
\begin{algorithmic}[1]
		\Procedure{$\texttt{IntOpt}$}{${\bf x}, c$}.
        \State ${\bf z} \gets \max\big(\big\lceil \frac{c-\sum_i x_i}{\text{dim}(\bf{x})}\big\rceil , -{\bf x}\big)$
        \State $t \gets \|\bf{z}\|_{\infty}$
        \State $I \gets \text{sorted indices of ${\bf x}$ in ascending order}$
        \State $j\gets 0$
        \While{$\sum_i z_i > c-\sum_i x_i$}
         \vspace{2pt}
            \State $\Delta \gets \sum_i z_i - c+\sum_i x_i$
             \vspace{2pt}
            \State $z_{I[j]} \gets \max(z_{I[j]}-\Delta, -x_{I[j]}, -t)$
             \vspace{2pt}
            \State $j \gets (j + 1) \mod |I|$
            \If {$j = 0$}
                \State $t \gets t + 1$
            \EndIf
        \EndWhile
        \State \bf{return } ${\bf x}+ {\bf z}$
        \EndProcedure
	\end{algorithmic}
\end{minipage}
\end{algorithm}
The detailed pseudocode of \texttt{InfTDA} is provided in Algorithm~\ref{alg: InfTDA}. The process begins by constructing the root $\tilde{V}_0$ of the differentially private tree in line 1. Here, $u_0$ denotes the root node of the input tree, while $n$ represents its attribute. The algorithm then starts the TopDown loop in line 2. 
Each iteration aims to construct the set of nodes at level $k$ of the differentially private tree, instantiated in line 3.
Each node of the previous level $k-1$ is sampled in line 4 and used as a constraint. 
In line 5 the discrete Gaussian mechanism with $\rho/d$ privacy budget (for zCDP) is applied to the attributes ${\bf q}_{\mathcal{C}}(u)$ of the child nodes of the constraint. 
Then, in line 6 the private attributes $\tilde{\bf{q}}$ are post-processed to satisfy the constraints. 
The algorithm \texttt{IntOpt} solves the following integer optimization problem by minimizing the Chebyshev distance
\begin{equation}
	\mathcal{P}({\bf x}, c) := \arg \min_{\bf{y}}||{\bf x}-{\bf y}||_{\infty} \quad\text{s.t.} \quad \sum_{i=1}^{\text{dim}({\bf x})}y_i = c.\quad \text{and} \quad y_{i}\in \mathbb{N}_{0} \quad \forall i\in[1,\dots,\text{dim}({\bf x})]
\end{equation}
In line 7, the set of child nodes of the constraint is constructed, and it is augmented with the corresponding post-processed DP attributes in line 8, dropping nodes with zero attributes. 
This final step effectively reduces the size of the DP tree and the running time of the algorithm, particularly for sparse datasets.
In fact, if there is a node $u_k$ with optimized attribute $\bar{q}(u_k)=0$, then, by consistency, the entire branch of the tree starting at $u_k$ will also have nodes with zero attributes.
Finally, in line 9 the set of DP nodes is updated, and this set will serve as constraints in the next iteration.
The algorithm outputs the leaves of the DP tree, which constitute a full contingency table $\tilde{D}$. 

\begin{theorem}[Utility and Privacy of \texttt{InfTDA}] 
\label{theorem: utility}
Let $D = (\mathcal{X}_1\times \mathcal{X}_d)^n$ be a dataset containing $n$ points of $d$ categorical values. Let $\tilde{D}$ be the dataset returned by \texttt{InfTDA}. Then, $|\tilde{D}| = n$ and for any $k\in [1, \dots, d]$ and $\beta \in (0,1)$ we have that
\begin{equation}
    \textnormal{Pr}\left[\max_{y\in \prod_{i=1}^{k}\mathcal{X}_i}\left|q_{S_{k},y}(D)-q_{S_k, y}(\tilde{D})\right|\leq \sum_{\ell= 1}^{k}\sqrt{\frac{8d}{\rho}\log\left(\frac{k\cdot\prod_{i=1}^{\ell}|\mathcal{X}_i|}{\beta}\right)}\right]\geq 1-\beta.
\end{equation}
Additionally, $\texttt{InfTDA}$ satisfies $\rho$-zCDP.
\end{theorem}
\begin{proof}
    We refer the reader to Appendix~\ref{app: utility of InfTDA} for the utility proof. Privacy is ensured through composition, parallel composition, and post-processing. A privacy budget of $\rho/d$ is allocated to each level of the tree, as the $\ell_2$-global sensitivity of each $k$-hierarchical query is $\sqrt{2}$, and each level represents a partition of the data. By applying the composition rule of zCDP across $d$ levels, the privacy guarantee is established.
\end{proof}

\paragraph{IntOpt} The main difference between \texttt{InfTDA} and the TopDown algorithm developed by the US Census \cite{Abowd20222020} lies in the optimization procedure. In \texttt{InfTDA}, Chebyshev minimization was chosen for two key reasons: to enable a utility analysis of the mechanism and to design a straightforward optimizer that operates directly in the integer domain. In contrast, in \cite{Abowd20222020} the mechanism solves a Euclidean minimization in two steps. First, the problem is solved in the real domain using convex optimization techniques; then it searches for the best rounding using an integer linear program. \texttt{IntOPT} solves Chebyshev minimization in one shot entirely in the integer domain. Algorithm \ref{alg: InfTDA} shows a simple version that runs in $O(c)$ steps. A version that runs in $O(\min(|\dim({\bf x})|^2, c))$ is presented in Appendix C in \cite{boninsegna2025}. 

The main approach to solve $\mathcal{P}(\mathbf{x}, c)$ is to first compute a lower bound for the Chebyshev distance. Based on this bound, a non-negative solution is then constructed with $\ell_\infty$ norm no greater than the lower bound and with elements summing to a value exceeding $c$. Subsequently, \texttt{IntOPT} reduces the elements to ensure that the summation constraint is satisfied, while increasing the Chebyshev distance as little as possible.
More precisely, $\mathcal{P}({\bf x}, c)$ is equivalent to the following problem in the real domain for $\bf{z} = y-x$
\begin{equation*}
    \min \alpha \geq 0 \quad \text{s.t. } z_i \in [-\alpha, \alpha] \quad \text{and}\quad z_i\geq -x_i \quad \text{and} \quad \sum_{j}z_j = c - \sum_j x_j \quad \forall i \in [1,\dots, \text{dim}(\bf x)].
\end{equation*}
This problem has a lower bound $\alpha \geq \max\left( \left|\frac{c-\sum_{i=1}x_i}{\text{dim}({\bf x})}\right|\,;\, -\min_{i} x_i \right)$ in the real domain, which translates to a lower bound in the integer domain by taking the ceiling function. In line 2 of \texttt{IntOPT}, $\bf z$ is initialized such that $\|{\bf z}\|_\infty\leq \max\left( \left\lceil\left|\frac{c-\sum_{i=1}x_i}{\text{dim}({\bf x})}\right|\right\rceil\,;\, -\min_{i} x_i \right)$. Furthermore, since $\sum_i z_i \geq c - \sum_i x_i$ \footnote{$\sum_{i=1}\max\bigg(\bigg\lceil \frac{c-\sum_{i=1} x_i}{\text{dim}({\bf x})}\bigg\rceil \,,\,-x_i\bigg)\geq  \sum_{i=1}\bigg\lceil \frac{c-\sum_{i=1} x_i}{\text{dim}({\bf x})}\bigg\rceil\geq c-\sum_{i=1}x_i$}, it is sufficient to reduce the elements of $\mathbf{z}$ until this constraint is satisfied. The reduction is applied uniformly across the elements to minimise the increase in Chebyshev distance. This procedure can be implemented through different strategies, resulting in multiple possible solutions. Indeed, Chebyshev optimisation does not admit a unique solution.
The method adopted in \texttt{IntOPT} is based on two main choices: (a) the reduction is applied iteratively following the indices of $\mathbf{x}$ that sort it in ascending order (line 4), and (b) during the first iteration, each element is reduced by the maximum possible amount that does not alter $\|\mathbf{z}\|_{\infty}$ (line 3).
The underlying idea is to prioritise setting to zero the private counts that are negative or small, as these are more likely to correspond to sensitive counts that were originally zero. Such choices have been demonstrated to effectively reduce the presence of false positives in the private contingency table \cite{boninsegna2025}, when compared to alternative optimisation strategies.

\section{Conclusions and Open Questions}
This paper has presented a generalisation of \texttt{InfTDA} for a generic collection of data points characterised by $d$ categorical features. The proposed mechanism offers theoretical guarantees for any $k$-hierarchical query, ensuring a maximum absolute error of $O(\sqrt{k^3 d})$. Furthermore, the mechanism has been shown to perform effectively in practice on both real-world and synthetic datasets \cite{boninsegna2025}. However, it remains an open question whether it is possible to achieve an error of $\tilde{O}(\sqrt{k d})$, as attained by the hierarchical mechanism, while still producing a contingency table $\tilde{D}$ with non-negative integer counts.

\newpage
\bibliographystyle{unsrt}  
\bibliography{references}  

\newpage
\appendix
\section{Utility of \texttt{InfTDA}}
\label{app: utility of InfTDA}
We start with the utility of the discrete Gaussian mechanism.
\begin{corollary}[Corollary 9 \cite{Canonne_Kamath_Steinke_2022}]
\label{corollary: discrete Gaussian tail}
Let $Z\sim \mathcal{N}_{\mathbb{Z}}(0, \sigma^2)$. Then $\textnormal{Var}[Z]\leq \sigma^2$ and $\textnormal{Pr}[Z\geq t]\leq e^{-\frac{t^2}{2\sigma^2}}$ for any $t\geq 0$.  
\end{corollary}
Now we are ready to prove the full version of the utility Theorem for \texttt{InfTDA}.
\begin{proof}[Proof of utility in Theorem \ref{theorem: utility}]
Let us focus on one $k$-hierarchical query, $q_{S_{k}, y}(D)$ for some $y \in \prod_{i=1}^k \mathcal{X}_i$ and $k \in [1, \dots, d]$. The non-negative hierarchical tree represents this query with an attribute of a node $u_k$ at level $k$. Let this attribute be $q(u_k)$, while $\Tilde{q}(u_k)$ be the attribute returned by the Gaussian mechanism before the optimization is applied, and $\Bar{q}(u_k)$ after the optimization. Since $\bar{q}(u_k)$ is the differentially private query returned by \texttt{InfTDA}, we have by triangle inequality that
    \begin{equation*}
        |\err(u_k)| = |q(u_k) - \Bar{q}(u_k)| \leq |q(u_k) -\Tilde{q}(u_k)| + |\Tilde{q}(u_k) - \Bar{q}(u_k)|.
    \end{equation*}
    The first term is just the absolute value of a Gaussian random variable, so we focus on the second term. Let $u_{k-1}$ be the father node of $u_k$, then $\Bar{q}(u_k)$ is an element of the vector ${\bf \bar{q}}_{\mathcal{C}}(u_{k-1})\in \mathbb{N}_0^{|\mathcal{X}_k|}$ solution to the optimization problem $\mathcal{P}({\bf \tilde{q}}_{\mathcal{C}}(u_{k-1}), c)$, where $c=\bar{q}(u_{k-1})$, thus
    \begin{equation*}
        |\tilde{q}(u_k)-\bar{q}(u_k)|\leq \|{\bf \tilde{q}}_{\mathcal{C}}(u_{k-1})-{\bf \bar{q}}_{\mathcal{C}}(u_{k-1})\|_\infty.
    \end{equation*}
     Let us consider another vector $\boldsymbol{\xi}\in \mathbb{Z}^{|\mathcal{X}_k|}$, called the \emph{offset}, such that ${\bf q}_{\mathcal{C}}(u_{k-1}) + {\boldsymbol{\xi}}$ lies within the feasible region of the constrained optimization problem $\mathcal{P}$, then
     \begin{align}
      \label{eq: proof constrain 1}
     q_{\mathcal{C}, j}(u_{k-1})+\xi_j &\geq 0\\
    \label{eq: proof constrain 2}
     \sum_{j=1}^{|\mathcal{X}_k|} q_{\mathcal{C}, j}(u_{k-1})+\xi_j &= \bar{q}(u_{k-1}).
    \end{align}
    As the vector ${\bf \bar{q}}_{\mathcal{C}}(u_{k-1})$ is a solution to the optimization problem, it minimizes the Chebyshev distance with ${\bf \tilde{q}}_{\mathcal{C}}(u_{k-1})$ under the non-negativity and summation constraints, then by triangle inequality
     \begin{align*}
         \|{\bf \tilde{q}}_{\mathcal{C}}(u_{k-1})-{\bf \bar{q}}_{\mathcal{C}}(u_{k-1})\|_\infty&\leq \|{\bf \tilde{q}}_{\mathcal{C}}(u_{k-1})-{\bf q}_{\mathcal{C}}(u_{k-1})-\boldsymbol{\xi}\|_\infty && \text{\{${\bf \bar{q}}_{\mathcal{C}}(u_{k-1})$\} is a \emph{optimal} solution\}}\\
         &\leq \|{\bf \tilde{q}}_{\mathcal{C}}(u_{k-1})-{\bf q}_{\mathcal{C}}(u_{k-1})\|_\infty + \|\boldsymbol{\xi}\|_\infty && \text{\{due to triangle inequality\}}.
     \end{align*}
     As $|q(u_k)-\tilde{q}(u_k)|\leq \|{\bf \tilde{q}}_{\mathcal{C}}(u_{k-1})-{\bf q}_{\mathcal{C}}(u_{k-1})\|_\infty$, the upper bound becomes
    \begin{equation}
    \label{eq: proof continue}
        |\err(u_k)|\leq 2 \|{\bf \tilde{q}}_{\mathcal{C}}(u_{k-1})-{\bf q}_{\mathcal{C}}(u_{k-1})\|_\infty + \|\boldsymbol{\xi}\|_\infty.
    \end{equation}
    The problem is now to find an upper bound for $\|\boldsymbol{\xi}\|_{\infty}$.

    To upper bound $\|\boldsymbol{\xi}\|_{\infty}$, we now construct an example of $\boldsymbol{\xi}$ satisfying the constraints and having a bounded $\ell_\infty$ norm. By construction, from Equation \ref{eq: proof constrain 2} and hierarchical consistency, we have that
    \begin{equation*}
        \sum_{j=1}^{|\mathcal{X}_k|}\xi_j = \bar{q}(u_{k-1})-\sum_{j=1}^{|\mathcal{X}_k|}q_{\mathcal{C}, j}(u_{k-1}) = \bar{q}(u_{k-1})-q(u_{k-1}) =\err(u_{k-1}).
    \end{equation*}
    If $\err(u_{k-1})\geq 0$ we can take $\xi_j = \frac{\err (u_{k-1})}{|\mathcal{X}_k|}$ for any $j\in[|\mathcal{X}_k|]$ as a solution satisfying the summation constraint and the inequality constraint in Equation~\ref{eq: proof constrain 1}. However, this is not sufficient. If $\err(u_{k-1})< 0$ the inequality constraint might not be satisfied. In this scenario we might consider a solution where $\xi_i = 0$ for any $i \in [|\mathcal{X}_k|]\setminus\{i^*\}$ where $\xi_{i^*} = -|\err(u_{k-1})|$. Any zero element satisfies the constraint in Equation~\ref{eq: proof constrain 1} as $q_{\mathcal{C}, j}(u_{k-1})\geq 0$. If $\xi_{i^*} \geq -q_{\mathcal{C}, i^*}(u_{k - 1})$ then we finish and obtain an upper bound $\|\boldsymbol{\xi}\|_{\infty}\leq |\err(u_{k-1})|$. In the other case where $\xi_{i^*}<-q_{\mathcal{C}, i^*}(u_{k-1})$ we need to augment $\xi_{i^*}$ up to meet $-q_{\mathcal{C}, i^*}(u_{k-1})$. By doing so we increase $\sum_i \xi_i$ making necessary to decrease some elements of the offset. As we are reducing elements that initially are zero, the new offset still contains only negative elements, and as  $\sum_i \xi_i = -|\err(u_{k-1})|$ any element cannot be less than $-|\err(u_{k-1})|$. Thus, we conclude that there always exists an offset such that $\|\boldsymbol{\xi}\|_{\infty}\leq |\err(u_{k-1})|$.

    Now we continue from the upper bound in Equation \ref{eq: proof continue}
    \begin{equation}
    \label{eq: branch}
        |\err(u_k)|\leq 2 \|{\bf \tilde{q}}_{\mathcal{C}}(u_{k-1})-{\bf q}_{\mathcal{C}}(u_{k-1})\|_\infty + |\err(u_{k-1})|.
    \end{equation}
    By taking the maximum we obtain
    \begin{equation*}
        \max_{u_k \in V_k}|\err(u_{k})|\leq 2\cdot \max_{u_{k-1} \in V_{k-1}} \|{\bf \tilde{q}}_{\mathcal{C}}(u_{k-1})-{\bf q}_{\mathcal{C}}(u_{k-1})\|_\infty + \max_{u_{k-1} \in V_{k-1}}|\err(u_{k-1})|.
    \end{equation*}
    Completing the recurrence relation by ending at $\err(u_0) = 0$ we get
    \begin{equation}
    \label{eq: proof InfTDA 2}
          \max_{u_{k} \in V_{k}}|\err(u_k)| \leq 2\cdot \sum_{\ell= 1}^{k}\max_{u_{k-\ell} \in V_{k-\ell}} \|{\bf \tilde{q}}_{\mathcal{C}}(u_{k-\ell})-{\bf q}_{\mathcal{C}}(u_{k-\ell})\|_\infty.
    \end{equation}
    Note that in the previous equation $u_{k-1}$ is the father node of $u_{k}$, as previously defined, whereas $u_{k-2}$ is the father node of $u_{k-1}$ and so on.
    We may upper bound each term in Equation \ref{eq: proof InfTDA 2} by using Corollary \ref{corollary: discrete Gaussian tail} with a union bound over the dimension of each vector, which is $|\mathcal{X}_{k-\ell +1}|$, times the number of nodes at level $V_{k-\ell}$, which are $\prod_{i =1}^{k-\ell}|\mathcal{X}_i|$. Thus, we obtain 
    \begin{equation*}
        \textnormal{Pr}\left[\max_{u_{k-\ell} \in V_{k-\ell}} \|{\bf \tilde{q}}_{\mathcal{C}}(u_{k-\ell})-{\bf q}_{\mathcal{C}}(u_{k-\ell})\|_\infty\leq \sqrt{\frac{2d}{\rho}\log\left(\frac{k\cdot\prod_{i=1}^{k-\ell + 1}|\mathcal{X}_{i}|}{\beta}\right)}\right]\geq 1-{\beta \over k}.
    \end{equation*}
    By re-indexing Equation \ref{eq: proof InfTDA 2}, summing the previous upper bounds, and applying an union bound, we finally obtain our claim
    \begin{equation*}
        \textnormal{Pr}\left[\max_{u_{k}\in V_k}|\err(u_k)| \leq \sum_{\ell= 1}^{k}\sqrt{\frac{8d}{\rho}\log\left(\frac{k\cdot\prod_{i=1}^{\ell}|\mathcal{X}_i|}{\beta}\right)}\right]\geq 1-\beta.
    \end{equation*}
\end{proof}
\end{document}